\documentclass[11pt]{article}
\usepackage{amsmath,amsthm,amssymb,cite,enumerate}
\usepackage[a4paper,left=0.5in,right=1in]{geometry}
\usepackage[usenames]{color}
\usepackage{graphicx}
\usepackage{epsfig}
\usepackage{dcolumn}
\usepackage{bm}
\usepackage{authblk} 

\begin{document}

\def \d {{\rm d}}

\def \bm #1 {\mbox{\boldmath{$m_{(#1)}$}}}

\def \bF {\mbox{\boldmath{$F$}}}
\def \bV {\mbox{\boldmath{$V$}}}
\def \bff {\mbox{\boldmath{$f$}}}
\def \bT {\mbox{\boldmath{$T$}}}
\def \bk {\mbox{\boldmath{$k$}}}
\def \bl {\mbox{\boldmath{$\ell$}}}
\def \bn {\mbox{\boldmath{$n$}}}
\def \bbm {\mbox{\boldmath{$m$}}}
\def \tbbm {\mbox{\boldmath{$\bar m$}}}

\def \T {\bigtriangleup}
\newcommand{\msub}[2]{m^{(#1)}_{#2}}
\newcommand{\msup}[2]{m_{(#1)}^{#2}}

\newcommand{\be}{\begin{equation}}
\newcommand{\ee}{\end{equation}}

\newcommand{\beqn}{\begin{eqnarray}}
\newcommand{\eeqn}{\end{eqnarray}}
\newcommand{\AdS}{anti--de~Sitter }
\newcommand{\AAdS}{\mbox{(anti--)}de~Sitter }
\newcommand{\AAN}{\mbox{(anti--)}Nariai }
\newcommand{\AS}{Aichelburg-Sexl }
\newcommand{\pa}{\partial}
\newcommand{\pp}{{\it pp\,}-}
\newcommand{\ba}{\begin{array}}
\newcommand{\ea}{\end{array}}

\newcommand{\M}[3] {{\stackrel{#1}{M}}_{{#2}{#3}}}
\newcommand{\m}[3] {{\stackrel{\hspace{.3cm}#1}{m}}_{\!{#2}{#3}}\,}

\newcommand{\tr}{\textcolor{red}}
\newcommand{\tb}{\textcolor{blue}}
\newcommand{\tg}{\textcolor{green}}

\def\a{\alpha}
\def\g{\gamma}
\def\de{\delta}

\def\b{{\kappa_0}}

\def\E{{\cal E}}
\def\B{{\cal B}}
\def\R{{\cal R}}
\def\F{{\cal F}}
\def\L{{\cal L}}

\def\e{e}
\def\bb{b}

\newtheorem{theorem}{Theorem}[section] 
\newtheorem{cor}[theorem]{Corollary} 
\newtheorem{lemma}[theorem]{Lemma} 
\newtheorem{proposition}[theorem]{Proposition}
\newtheorem{definition}[theorem]{Definition}
\newtheorem{remark}[theorem]{Remark}

\title{Universal electromagnetic fields}

\author[1]{Sigbj\o rn Hervik\thanks{sigbjorn.hervik@uis.no}}
\author[2,3]{Marcello Ortaggio\thanks{ortaggio(at)math(dot)cas(dot)cz}}
\author[2]{Vojt\v ech Pravda\thanks{pravda@math.cas.cz}}

\affil[1]{Faculty of Science and Technology, University of Stavanger, N-4036 Stavanger, Norway}
\affil[2]{Institute of Mathematics of the Czech Academy of Sciences, \newline \v Zitn\' a 25, 115 67 Prague 1, Czech Republic}
\affil[3]{Instituto de Ciencias F\'{\i}sicas y Matem\'aticas, Universidad Austral de Chile, \newline Edificio Emilio Pugin, cuarto piso, Campus Isla Teja, Valdivia, Chile}

\maketitle

\abstract{We study {\em universal} electromagnetic (test) fields, i.e., $p$-forms fields $\bF$ that solve simultaneously (virtually) any generalized electrodynamics (containing arbitrary powers and derivatives of $\bF$ in the field equations) in $n$ spacetime dimensions. One of the main results is a sufficient condition: any null $\bF$ that solves Maxwell's equations in a Kundt spacetime of aligned Weyl and traceless-Ricci type III is universal (in particular thus providing examples of $p$-form Galileons on curved Kundt backgrounds). In addition, a few examples in Kundt spacetimes of Weyl type II are presented. Some necessary conditions are also obtained, which are particularly strong in the case $n=4=2p$: all the scalar invariants of a universal 2-form in four dimensions must be constant, and vanish in the special case of a null $\bF$.}

\vspace{.2cm}


\tableofcontents

\section{Introduction}

\label{intro}

\subsection{Background}

Modifications of Maxwell's equations have been proposed in the context of classical theories in order to cure the divergent electron's self-energy \cite{Mie12,Born33,BorInf34,Bopp40,Podolsky42}. They have also appeared in effective theories derived from quantum electrodynamics (QED) \cite{HeiEul36,Weisskopf36,Euler36,Schwinger51}) or from string theory (cf., e.g., the review \cite{Tseytlin00} and references therein). A particularly well-known example of a generalized theory is given by non-linear electrodynamics (NLE), which is defined by a Lagrangian depending (in principle arbitrarily) on the two algebraic invariants $F_{ab}F^{ab}$ and $F_{ab}*F^{ab}$ \cite{Born37,Plebanski70}.

In early works by Schr\"odinger \cite{Schroedinger35,Schroedinger43} it was observed that all {\em null} fields (defined by $F_{ab}F^{ab}=0=F_{ab}*F^{ab}$) which solve Maxwell's equations also automatically solve the field equations of any NLE (in vacuum). In this sense, null fields display theory-independent properties. Subsequently, it was noticed that {\em plane waves} (a special case of null fields) solve not only NLE but also higher-derivative theories in a flat spacetime \cite{Deser75}. It is clearly desirable to understand to what extent these results can be generalized, e.g., to field configurations other than plane-waves, as well as to curved backgrounds. Recently, a larger class of {\em universal} solutions
has been constructed \cite{OrtPra18} -- it consists of null electromagnetic fields in four-dimensional Kundt spacetimes of (aligned) Petrov and traceless-Ricci type III. In the present paper we will further extend the results \cite{OrtPra18}, particularly in two directions. First, we shall consider $p$-form fields in arbitrary dimensions $n$, which are of more direct interest for supergravity and string theory applications (recovering the sufficient conditions for universality of \cite{OrtPra18} in the special case $n=4=2p$). Additionally, we will also discuss the {\em non-null} case and also obtain certain necessary conditions for a $p$-form to be universal. In this case, more progress will be possible when $n=4=2p$.

In the rest of this section we provide some preliminary definitions. Section~\ref{sec_sufficient} presents a few sufficient conditions for a  $p$-form field to be universal (or 0-universal, as defined below) -- along with other results, it contains Theorem~\ref{prop_UVSI}, one of the main conclusions of this paper. In section~\ref{sec_necessary} we focus instead on necessary conditions. In particular, more progress is possible in the case of 2-forms in four dimensions, for which we obtain Theorem~\ref{prop_univ_necess_4D} and Proposition~\ref{prop_null_univ_necess_4D}. Appendix~\ref{app_D_Kundt} summarizes the notation used throughout the paper and contains some auxiliary technical results (some of which already known, cf. the references given there), needed in the proofs of the main theorems and propositions. 
Appendix~\ref{app_subsec_balanced} also contains some additional related results (Lemmas~\ref{lemma_1balanced} and \ref{lemma_1deriv}) of some interest for future studies.

\subsection{Preliminaries and definitions}

Throughout the paper we will restrict ourselves to the case in which there are no charges and currents, therefore we will omit the word ``sourcefree'' when referring to the (generalized) Maxwell equations. We will consider only Maxwell {\em test} $p$-form fields in an $n$-dimensional spacetime (i.e., we will not consider the back-reaction). We will denote by $d$, $\delta$ and $*$, respectively, the exterior derivative, the codifferential and the Hodge dual of differential forms.

Let us start with a definition which will be used throughout this work.
\begin{definition}[Universal electromagnetic field]
\label{def_UM}
	A $p$-form $\bF$ is called universal if it satisfies the pair of generalized Maxwell's equations
	\be
	 \d\bF=0,  \qquad *\d\!*\!\!\tilde\bF=0 , 
	 \label{UM}
	\ee
	where the second equation holds simultaneously for {\em all} $p$-forms $\tilde\bF$ constructed polynomially from $\bF$ and its covariant derivatives (hereafter, it is understood that the dual $(n-p)$-form $*\bF_{b_1\ldots b_{n-p}}\equiv\frac{1}{p!}\epsilon^{a_1\ldots a_p}_{\qquad \ b_1\ldots b_{n-p}}F_{a_1\ldots a_p}$ and its derivatives can also be used in the construction).
\end{definition}

In Definition~\ref{def_UM}, $\tilde\bF$ may thus include terms with arbitrary higher-order derivative ``corrections''. Furthermore, the scalar coefficients appearing in such a polynomial need {\em not} be polynomials of the scalar invariants of $\bF$ (and their covariant derivatives).
Of course, a universal $\bF$ in particular satisfies the standard Maxwell equations (for the special choice  $\tilde\bF=\bF$).\footnote{Note that there is an ``asymmetry'' in the above definition in the sense that $\bF$ must be closed and co-closed, whereas all the $\tilde\bF$ one can construct are only required to be co-closed. In the special case $n=2p$, however, any possible $\tilde\bF$ must also be closed (since one can replace $\tilde\bF$ by $\hat\bF\equiv *\bF$ in Definition~\ref{def_UM} and use the fact that $\hat\bF$ must be co-closed).} {Definition~\ref{def_UM} extends to arbitrary $n$ and $p$ a definition proposed in \cite{OrtPra18} for $n=4=2p$.}

\begin{remark}[Non-linear electrodynamics (NLE)]
	NLE is usually defined in the case $n=2p=4$ (cf., e.g., \cite{Plebanski70}) by a Lagrangian which depends only on the two 0-order (i.e., algebraic) invariants $F_{ab}F^{ab}$ and $F_{ab}{}^*F^{ab}$. Hence, in NLE {the 2-form} $\tilde\bF$ is a linear 		  
	combination $\tilde\bF=\alpha\bF+\beta{}^*\bF$, where $\alpha$ and $\beta$ are (virtually arbitrary) functions of $F_{ab}F^{ab}$ and $F_{ab}{}^*F^{ab}$. It follows that a universal $\bF$ also obeys, in particular, the field equations of any NLE (e.g., Born-Infeld's theory).
	\label{rem_NLE}
\end{remark}

{One could also define extensions of NLE to arbitrary $n$ and $p$ by considering theories where $\tilde\bF$ is constructed algebraically from $\bF$ (i.e., without taking covariant derivatives).\footnote{In general, there can be more than two {independent 0-order invariants that can be used to construct a Lagrangian}. For example, a 2-form in $n$ dimensions possesses $[n/2]$ algebraic invariants (cf. section~4 of \cite{Born37}).} Since these} represent a very special subclass of theories of electrodynamics, it appears useful to define also the following more restricted notion of universality.

\begin{definition}[0-universal electromagnetic field]
\label{def_0UM}
	A $p$-form $\bF$ is called 0-universal if it satisfies the pair of generalized Maxwell's equations~\eqref{UM}, where $\tilde\bF$ can be any $p$-form constructed {\em algebraically} and polynomially from $\bF$.
\end{definition}
In particular, 0-universal $p$-forms solve theories for which $\tilde\bF$ is of the form ``$\bF$+higher powers of $\bF$''.

Similarly, it is also useful to define 
\begin{definition}[$K$-universal electromagnetic field]
\label{def_KUM}
	A $p$-form $\bF$ is called K-universal if it satisfies the pair of generalized Maxwell's equations~\eqref{UM}, where $\tilde\bF$ can be any $p$-form constructed polynomially from $\bF$ and its first $K$ covariant derivatives.
\end{definition}

In particular, 1-universal $p$-forms in flat space provide examples of $p$-form Galileons \cite{DefDesEsp10}.

The notion of {\em null} fields is well known for 2-forms in four dimensions (cf., e.g., \cite{Stephanibook}); this has been extended to arbitrary dimension~$n$ \cite{Sokolowskietal93} and rank~$p$ \cite{OrtPra16}. {Let us thus recall}
\begin{definition}[Type N (or null) $p$-forms \cite{OrtPra16}]
	\label{def_N}
 At a spacetime point, a $p$-form $\bF$ is of type N if it satisfies
\be
	\ell^a F_{a b_1\ldots b_{p-1}}=0 , \qquad \ell_{[a}F_{b_1\ldots b_{p}]}=0 ,
	\label{BelDeb}
\ee
where (\eqref{BelDeb} implies that) $\bl$ is a null vector. Equivalently, the second condition can be  replaced by $\ell^a\, {}^*F_{a b_1\ldots b_{n-p-1}}=0$.
\end{definition}

 A further useful definition is the following. 
\begin{definition}[CSI and VSI electromagnetic field]
\label{def_CSI}
	A $p$-form $\bF$ in a spacetime with metric $g_{ab}$ is called CSI (``constant scalar invariants'') if the scalar polynomial invariants constructed from $\bF$ and its covariant derivatives of arbitrary order are constant. If all such invariant vanish, $\bF$ is VSI (``vanishing scalar invariants'') \cite{OrtPra16}. If only derivatives up to order $K$ are considered, we will use the notation CSI$_K$ and VSI$_K$.
\end{definition}

\paragraph{Notation} Throughout the paper we will employ a null frame in an $n$-dimensional spacetime and the corresponding Ricci rotation coefficients and directional derivatives, as defined in Appendix~\ref{app_sub_not}. We will use the abstract tensor notation ``Riem'' to denote the Riemann tensor. A tensor defined by the part of boost weight (b.w.) $b$ of the Riemann tensor or of its covariant derivative will be denoted, respectively, as Riem$_b$ or $[\nabla\mbox{(Riem)}]_b$; by contrast, a frame component of b.w. $b$ of the Riemann tensor will be denoted as ${\cal R}_b$. For general information about the b.w. classification of tensors we refer to \cite{Milsonetal05} and, e.g., to the review \cite{OrtPraPra13rev}.

\section{Sufficient conditions}

\label{sec_sufficient}

\subsection{Covariantly constant $p$-forms}

An observation made in \cite{OrtPra18} in the case $n=4=2p$ can be extended to arbitrary $n$ and $p$ as follows:
\begin{proposition}
\label{prop_covcost}
	A covariantly constant $p$-form $\bF$ is universal and CSI.
\end{proposition}

\begin{proof} 

Since any covariant derivative of $\bF$ is zero, it is obvious that also $\d\bF=0=\delta\bF$ . Similarly, any scalar invariant constructed from $\bF$ must be constant (in particular, those constructed from the derivatives of $\bF$ are zero). Any $\tilde\bF$ is clearly also covariantly constant and thus has the same properties.

\end{proof} 

For example, in a flat spacetime in Cartesian coordinates any wedge-product of the coordinate differentials gives a covariantly constant form. More generally, it is well-known that in direct product spacetimes, the volume form of each of the factor spaces is covariantly constant (as follows easily from the results of \cite{Ficken39}).

\begin{remark}[Covariantly constant 2-forms in four dimensions] 
	\label{rem_cc4D}
A four-dimensional spacetime admitting a covariantly constant 2-form $\bF$ is {either} (cf. section~35.1.2 of \cite{Stephanibook} and references therein): (i) a direct product of two two-dimensional spacetimes, if $\bF$ is {\em non-null}; (ii) an aligned \pp wave of Riemann type~N, if $\bF$ is {\em null}. In both cases the spacetime is Kundt and null-recurrent.  
\end{remark}

\begin{remark}[Covariantly constant null $p$-forms] 
A covariantly constant {\em null} $\bF$ is obviously VSI (cf. also \cite{OrtPra16}). Moreover, the background spacetime must be an aligned \pp wave (which follows from $F_{ac_1\ldots c_{p-1}} {F_b}^{c_1\ldots c_{p-1}}\propto \ell_a\ell_b$ being covariantly constant, similarly as in \cite{Stephanibook}). This partly extends Remark~\ref{rem_cc4D} to arbitrary $n$ and $p$.
\end{remark}

\subsection{Null $p$-forms}

As mentioned in section~\ref{intro}, it was already known to Schr{\"o}dinger \cite{Schroedinger35,Schroedinger43} that all null Maxwell fields solve the equations for the electromagnetic field in any NLE. Generalizing \cite{OrtPra18}, this observation can be straightforwardly extended to all null (i.e., VSI$_0$ \cite{OrtPra16}) $p$-forms, so that (recalling Definition~\ref{def_0UM})

\begin{proposition}[Sufficient conditions for a 0-universal $\bF$]
\label{prop_0Unull}
	A null $p$-form $\bF$ that solves Maxwell's equations is 0-universal. 
\end{proposition}
{
\begin{proof} 
Since $\bF$ is null, any $p$-form $\tilde\bF$ constructed algebraically and polynomially out of $\bF$ can only be linear in $\bF$ (by boost-weight counting), and therefore obeys Maxwell's equations if $\bF$ does. 
\end{proof} 
}

We observe that no {direct} restrictions on the background spacetimes are placed by Proposition~\ref{prop_0Unull}, as opposed to Proposition~\ref{prop_UVSI} given below in the case of certain fully universal solutions. 

Recall that, in four dimensions, null 2-forms solutions to the Maxwell equations can be associated with shearfree cogruences of null geodesics via the Mariot-Robinson theorem (cf., e.g., \cite{Stephanibook}). A partial extension to $p$-forms in arbitrary dimensions has been discussed in \cite{Durkeeetal10} (see also \cite{Sokolowskietal93,Ortaggio07} for related earlier results). It is worth mentioning that, from a physical viewpoint, null fields are interesting since they characterize electromagnetic plane waves \cite{Schwinger51,syngespec}, the asymptotic behaviour of radiative systems \cite{penrosebook2}, and the field produced by high-energy sources \cite{Bergmann42,syngespec,RobRoz84}. They are also relevant to Penrose's limits in supergravity \cite{Guven87}.

\subsection{VSI $p$-forms}

\subsubsection{Sufficient condition}

The result of Proposition~\ref{prop_0Unull} for null forms can be considerably strengthened (i.e., to full universality) if one restricts oneself to null forms in suitable background spacetimes, namely
\begin{theorem}[Sufficient conditions for a universal $\bF$]
\label{prop_UVSI}
	In a Kundt spacetime of {aligned} Weyl and traceless-Ricci type III, any aligned null $p$-form $\bF$ that solves Maxwell's equations is universal {(and VSI)}.
\end{theorem}

\begin{proof} 
	
 Before starting, let us note that the considered background spacetime is necessarily degenerate Kundt and $DR=0=\delta_i R$ (cf. Propositions~\ref{prop_Kundt_deg} and \ref{prop_Kundt_III}). This implies that $\bF$ is VSI (Theorem~1.5 of \cite{OrtPra16}).

By Proposition~\ref{prop_Kundt_III}, $\nabla^{(I)}$(Riem) is balanced (and thus of type III or more special) for any $I\ge1$. This will be useful in the following. By the assumptions, the Riemann tensor can instead possess also b.w. 0 components, but is restricted to the special form 
\be
 R_{abcd}=\frac{2R}{n(n-1)}g_{a[c}{g_{d]b}}+\left(\mbox{b.w.}<0\right) ,
 \label{riem_spec}
\ee
where we have indicated the arbitrary components of negative b.w. implicitly in brackets, since their explicit form is not needed here. 

Bearing in mind the above observations and under the assumed conditions, the strategy is thus to prove that any possible $\tilde\bF$ is divergencefree.\footnote{Recall that for a $p$-form $\mbox{\boldmath{$\omega$}}$ the identity $\mbox{div}\,\omega_{a_1\ldots a_{p-1}}\equiv\omega^b_{\ \, a_1\ldots a_{p-1};b}=\mbox{sign}(g)(-1)^{n(p+1)}*\d*\!\omega_{a_1\ldots a_{p-1}}\equiv -\delta\omega_{a_1\ldots a_{p-1}}$ holds ($\mbox{sign}(g)=-1$ in this paper; our conventions for the dual of a $p$-form are given in Definition~\ref{def_UM})}

Now, since $\bF$ is VSI, any tensor constructed out of it and its covariant derivatives is also VSI and aligned (cf. Lemma B.4 of \cite{OrtPra16}) and thus contains only components of negative b.w. -- for totally antisymmetric tensors, such as $\tilde\bF$, the only possible b.w. is $-1$.
 It is thus clear that $\tilde\bF$ can only be {\em linear} in $\bF$ and its derivatives (higher order terms would have b.w. $-2$ or less). In general $\tilde\bF$ will thus be of the form
\be
	\tilde\bF=c_0\bF+c_I\left[\nabla^I\bF\right]_p+d_J\left[\nabla^J{}^*\bF\right]_p ,
	\label{tF}
\ee
where there is summation over the positive integers $I$ and $J$, the coefficients $c_0$, $c_I$ and $d_J$ are constants (recall that $\bF$ does not possess any non-zero invariant), and the notation $[\ldots]_{p}$ means that the quantity within square brackets needs to be contracted and/or antisymmetrized in such a way as to produce a $p$-form (the particular way how this is done is not important for our discussion). Let us now discuss separately the three possible types of terms in \eqref{tF} and show that they are all necessarily divergencefree.

The first term in \eqref{tF} is harmless since $\bF$ is itself divergencefree (and similarly for a term $d_0{}^*\bF$ that should be added to \eqref{tF} in the special case $n=2p$). 

The terms $\left[\nabla^I\bF\right]_p$ in \eqref{tF} need $I/2$ contractions with the metric tensor (so that $I$ must in fact be even) in order to produce an object of rank $p$, thus resulting in terms of one of the two forms 
\be
  g^{ab}\nabla^L_{\cdots}\nabla_b\nabla^{I-L-1}_{\cdots}F_{a\ldots} , \qquad g^{ab}\nabla^L_{\cdots}\nabla_b\nabla^K_{\cdots}\nabla_a\nabla^{I-L-K-2}_{\cdots}F_{\ldots} ,
	\label{F_partial}
\ee
where the ellipsis indicate implicitly all the remaining indices (contracted or not, and properly antisymmetrized -- this is not important four our analysis, as will be clear below). Now, if we considered theories in flat space, the terms \eqref{F_partial} could be rewritten as (note the respective indices contracted with $g^{ab}$)
\be
  g^{ab}\nabla^{I-1}_{\cdots}\nabla_b F_{a\ldots} , \qquad g^{ab}\nabla^{I-2}_{\cdots}\nabla_b\nabla_aF_{\ldots} .
	\label{F_partial2}
\ee
since covariant derivatives commute when the curvature is zero. (The advantage of the form \eqref{F_partial2} is that it can be easily handled with, as we shall discuss below after \eqref{Ft_simplified} also in the presence of curvature). However, this is generically not true in a curved spacetime, and if we want to shift indices of covariant derivatives we need to consider the generalized Ricci identity. For any tensor $\bT$, this can be expressed schematically as
\be
	[\nabla,\nabla]\bT=\bT\cdot\mbox{Riem} .
	\label{gen_ricci}
\ee

The idea is now to use this in order to show how in a curved background one can still arrive at terms of the form \eqref{F_partial2}, plus some ``corrections'' that turn out not to be an obstacle for our purpose (as we shall explain). 

Recalling that $\bF$ is VSI and all the covariant derivatives of the Riemann tensor are balanced, from \eqref{gen_ricci} we obtain
\be
	\nabla^L[\nabla,\nabla]\nabla^{I-L-2}\bF=\nabla^{I-2}\bF\cdot{\mbox{Riem}_0}+\left(\mbox{b.w.}<-1\right) .
	\label{ricci_F}
\ee
However, using \eqref{riem_spec} and recalling that a $p$-form cannot have components of b.w. smaller than $-1$, it is not difficult to see that (again schematically)
\be
 \left[\nabla^{I-2}\bF\cdot {\mbox{Riem}_0}+\left(\mbox{b.w.}<-1\right)\right]_p= R\left[\nabla^{I-2}\bF\right]_p . \label{eqcomm_form}
\ee
From $\delta_i R=0$ and $DR=0$ it follows that $R_{;a} \propto \ell_a$. Since all terms in \eqref{tF} are of b.w. $-1$, terms containing $R_{;a}$ do not
contribute to the divergence of $\tilde\bF$ and will be omitted from the next discussions. 

Equation \eqref{eqcomm_form} means that by repeated use of \eqref{ricci_F}, in terms $\left[\nabla^I\bF\right]_p$ we can shift a covariant derivative to any desired position, up to producing ``correction terms'' of the type $\left[\nabla^{I-2}\bF\right]_p$, i.e., of the form of the original terms but with order of differentiation reduced by 2. In turn, these can be reduced to the form \eqref{F_partial2} (with $I$ replaced by $I-2$), plus ``corrections'' of the type $\left[\nabla^{I-4}\bF\right]_p$, and so on. Eventually, the terms $\left[\nabla^I\bF\right]_p$ in \eqref{tF} are reduced to a sum of various terms of the form (recall that $I$ is even and therefore in the last step we obtain terms proportional to $\bF$)
\be
  \bF , \qquad g^{ab}\nabla^{M}_{\cdots}\nabla_b F_{a\ldots} , \qquad g^{ab}\nabla^{M-1}_{\cdots}\nabla_b\nabla_aF_{\ldots} \qquad (M=I-1, I-3, \ldots, 1) .
	\label{Ft_simplified}
\ee

We already observed that the first of these is divergencefree. The second one is (a derivative of) the divergence of $\bF$ and therefore vanishes. To study the third term it is useful to recall the Weitzenb\"ock identity for a $p$-form $\mbox{\boldmath{$\omega$}}$  \cite{Weitzenbock23}
\be
  g^{ab}\nabla_b\nabla_a \omega_{c_1\ldots c_p}=-\Delta\omega_{c_1\ldots c_p}+{p}R_{a[c_1}\omega^a_{\ \, c_2\ldots c_p]}-\frac{{p}(p-1)}{2}R_{ab[c_1c_2}\omega^{ab}_{\ \ \, c_3\ldots c_p]} , 
	\label{Weitzenbock}
\ee
where $\Delta=\d\delta+\delta\d$ is the Laplace-de Rham operator acting on forms (not to be confused with the symbol $\T$ defined in \eqref{covder}). Thanks to the fact that $\bF$ has only components of b.w. $-1$ and $\Delta \bF=0$ (since $\bF$ is harmonic), using \eqref{riem_spec} from \eqref{Weitzenbock} we obtain 
\be
  g^{ab}\nabla_b\nabla_a F_{c_1\ldots c_p}= \frac{p(n-p)}{n(n-1)}RF_{c_1\ldots c_p} .
\ee
Thus the third term in \eqref{Ft_simplified} is proportional to $\nabla^{M-1}_{\cdots}\bF$ (and thus of the form \eqref{F_partial}, with $I$ replaced by $M-1\le I-2$) and again can be manipulated iteratively until one obtains either $\bF$ or $g^{ab}\nabla^{M'}_{\cdots}\nabla_b F_{a\ldots}=0$ (for some $M'$). Therefore, also this term can produce only divergencefree quantities, as we wanted to prove.

Finally, let us discuss the terms $\left[\nabla^J{}^*\bF\right]_p$  in \eqref{tF}. There one needs $(n-2p+J)/2$ contractions in order to produce an object of rank $p$, with $J$ having the same parity of $n$ (and with the further condition $J>2p-n$ in the case $2p>n$; note that the case $J=2p-n$ requires no contractions, however it leads just to vanishing terms due to the total antisymmetrization and $\d\!*\!\!\bF=0$, and therefore needs no discussion). Similarly as above, we conclude that these will result in quantities of the form 
\be
  g^{ab}\nabla^L_{\cdots}\nabla_b\nabla^{J-L-1}_{\cdots}{}^*F_{a\ldots} , \qquad g^{ab}\nabla^L_{\cdots}\nabla_b\nabla^K_{\cdots}\nabla_a\nabla^{J-L-K-2}_{\cdots}{}^*F_{\ldots} ,
	\label{*F_partial}
\ee
which can be treated as done for \eqref{F_partial} ($*\bF$ is also divergencefree), thus completing the proof.

\end{proof}

\begin{remark}
	We observe that the above proof would be greatly simplified if we required $R=0$ as an additional assumption. Indeed, \eqref{ricci_F} would become simply $\nabla^L[\nabla,\nabla]\nabla^{I-L-2}\bF=\left(\mbox{b.w.}<-1\right)$. Therefore covariant derivatives of $\bF$ would essentially commute (i.e., the commutators would only produce terms of b.w. less than $-1$, which are set to zero by the ``operation'' $[\ldots]_p$) and the proof would thus be complete after \eqref{F_partial} (apart from taking into account the terms with $*\bF$).
\end{remark}

\begin{remark}[Examples]
	Note that, in particular, Ricci flat and Einstein {Kundt spacetimes} of Weyl typee III/N/O are permitted backgrounds for a universal VSI $\bF$.\footnote{Such spacetimes belong to the class of VSI/CSI metrics \cite{Coleyetal04vsi,ColHerPel06}, respectively (and include Minkowski and (A)dS). These have found various applications beyond general relativity. For example, in the context of type IIB supergravity, some VSI spacetimes coupled to null $p$-forms have been discussed in \cite{Coleyetal07}. In particular, the role of VSI \pp waves in the context of supergravity and string theory has been known for some time, see, e.g., \cite{KowalskiGlikman84,Hull84,Guven87,AmaKli89,HorSte90,Tseytlin93,BerKalOrt93}, also in the presence of supersymmetry, c.f. \cite{Tod83,KowalskiGlikman84,Hull84,Guven87,BerKalOrt93,Gauntlettetal03} and references therein.} The result of Proposition~\ref{prop_UVSI} was announced (without a proof) in section~2.4 of \cite{OrtPra16} and generalizes results obtained in \cite{OrtPra18} in the special case $n=4=2p$ (the above proof is also considerably more detailed than the one briefly sketched in \cite{OrtPra18}). An example in a Petrov type III spacetime in four dimensions (cf. eq.~(31.40) in \cite{Stephanibook}) is given by \cite{OrtPra16}
\beqn
 & & \d s^2 =2\d u\left[\d r+\frac{1}{2}\left(xr-xe^x\right)\d u\right]+ e^x(\d x^2+e^{2u}\d y^2) , \label{Petrov_EM} \nonumber \\
 & & \bF=e^{x/2}c(u)\d u\wedge\left(-\cos\frac{ye^u}{2}\d x+e^u\sin\frac{ye^u}{2}\d y\right) . \nonumber 
\eeqn
\end{remark}

\begin{remark}[Chern-Simons term]
	One can also consider generalized electrodynamics in the presence of a Chern-Simons term -- i.e., when (the dual of) the second of \eqref{UM} is replaced by $\d^*\tilde\bF+\alpha\tilde\bF\wedge\ldots\wedge\tilde\bF=0$, where $\alpha\neq0$ is an arbitrary constant, the second term of the equation contains $k$ factors $\tilde\bF$, and the corresponding number of spacetime dimensions is given by $n=p(k+1)-1$ (such modifications of Maxwell's equations take place, e.g, in the bosonic sector of minimal supergravity in five and eleven dimensions -- cf., e.g., \cite{Ortinbook} and references therein). The universal solutions of Proposition~\ref{prop_UVSI} thus clearly also solve such theories provided $k\ge2$, since $\tilde\bF\wedge\tilde\bF=0$ for a $\tilde\bF$ of type N. On the contrary, the special case $k=1$ results in a linear theory 
with $\d^*\tilde\bF+\alpha\tilde\bF=0$, which {\em cannot} be solved by a solution of \eqref{UM}. See also, e.g., \cite{FigPap01} for related results in a special case.

\end{remark}

\subsubsection{Solutions in adapted coordinates}

The result of Theorem~\ref{prop_UVSI} means that the spacetime metric can be written as 
\be
 \d s^2 =2\d u\left[\d r+H(u,r,x)\d u+W_\alpha(u,r,x)\d x^\alpha\right]+ g_{\alpha\beta}(u,x) \d x^\alpha\d x^\beta , 
 \label{Kundt_gen}
\ee
where $\bl=\partial_r$ is the Kundt vector, $\alpha,\beta=2 \dots n-1$, $x$ denotes collectively the set of coordinates $x^\alpha$, $g_{\alpha\beta}$ is an $(n-2)$-dimensional Riemannian metric of {\em constant curvature}
and
\beqn
 & & W_{\alpha}(u,r,x)=rW_{\alpha}^{(1)}(u,x)+W_{\alpha}^{(0)}(u,x) , \label{deg_Kundt1} \\
 & & H(u,r,x)=r^2H^{(2)}(u,x)+rH^{(1)}(u,x)+H^{(0)}(u,x) . \label{deg_Kundt2}
\eeqn
Further constraints following from the type III curvature conditions are (cf. the Riemann tensor components given in section~4.1 of \cite{ColHerPel06})
\beqn
 & & 2H^{(2)}=\frac{1}{2(n-2)}\left(W^{(1)\a}_{\ \ \ \ \ ||\a}+\frac{n-3}{2}W^{(1)\a}W^{(1)}_\a \right) , \label{H2} \\
 & & \R=\frac{n-3}{2}\left(W^{(1)\a}_{\ \ \ \ \ ||\a}-\frac{1}{2}W^{(1)\a}W^{(1)}_\a \right) , \label{RR} \\
 & & W_{\alpha||\beta}^{(1)}-\frac{1}{2}W_{\alpha}^{(1)}W_{\beta}^{(1)}=\frac{1}{n-2}\left(W^{(1)\g}_{\ \ \ \ \ ||\g}-\frac{1}{2}W^{(1)\g}W^{(1)}_\g \right)g_{\alpha\beta} , \\
 & & W_{\alpha}^{(1)}={\cal W}_{,\a} , \label{W} 
\eeqn
where $\R=\R(u)$ is the Ricci scalar of $g_{\alpha\beta}$ and ${\cal W}={\cal W}(u,x)$. Differentiating~\eqref{RR} w.r.t. $x^\beta$ gives $W^{(1)\a}_{\ \ \ \ \ ||\a\beta}-W^{(1)\a}W^{(1)}_{\a||\beta}=0$. In the special case when $\bl$ is recurrent (i.e., $W_{\alpha}^{(1)}=0$), eqs.~\eqref{H2}--\eqref{W} reduce to $H^{(2)}=0$ and $g_{\alpha\beta}$ must be flat.

Then any $p$-form of the form
\be
 \bF=\frac{1}{(p-1)!}f_{\alpha_1\ldots\alpha_{p-1}}(u,x)\d u\wedge\d x^{\alpha_1}\wedge\ldots\wedge\d x^{\alpha_{p-1}} ,
 \label{F_N_coords}
\ee
solves identically the equations~\eqref{UM} provided
\be
   f_{[\a_2\ldots\a_{p-1},\a_1]}=0 , \qquad  (\sqrt{\tilde g}\,f^{\beta\a_1\ldots\a_{p-2}})_{,\beta}=0 , 
	\label{Maxwell}
\ee
where $\tilde g\equiv\det g_{\alpha\beta}=-\det g_{ab}\equiv -g$ (effectively, eqs.~\eqref{Maxwell} are the Maxwell equations for the $(p-1)$-form $\bff$ in the $(n-2)$-dimensional Riemannian geometry associated with $g_{\alpha\beta}$).

\subsection{Universal $p$-forms in direct product spacetimes}

\label{subsec_dirprod}

One can take direct products of a spacetime permitted by Theorem~\ref{prop_UVSI}, say $\d s_1^2$, with any Riemannian space $\d s_2^2$ and obtain universal $p$-forms in a more general background $\d s^2=\d s_1^2+\d s_2^2$, provided $\bF$ still lives in $\d s_1^2$ (explicitly, this means that $\d s_1^2$ is given by~\eqref{Kundt_gen}, while $\d s_2^2=h_{AB}(y)\d y^A\d y^B$, with $h_{AB}(y)$ arbitrary, and the rest is unchanged). The spacetime $\d s^2$ is still degenerate Kundt but, in this case, its Weyl and traceless-Ricci types can also be II. A simple example with $p=2$ and $n=6$ is given by
\beqn
 & & \d s^2=2\d u[\d r+H(u,\zeta,\bar\zeta)\d u]+2\d\zeta\d\bar\zeta+a^2(\d\theta^2+\sin^2\theta\d\phi^2) , \nonumber \\
 & & \bF=\d u\wedge[f(u,\zeta)\d\zeta+\bar f(u,\bar\zeta)\d\bar\zeta] , \nonumber
\eeqn
where $H$ and $f$ are arbitrary functions of their arguments and $a$ a constant.

\subsection{Universal 2-forms in 4-dimensional Kundt type II spacetimes}

Another class of examples can be found for  four-dimensional Kundt spacetimes of Weyl/Ricci type II. Consider the type II Kundt spacetime (gravitational wave in a type D electrovac background \cite{GarAlvar84,Khlebnikov86,Lewand92,Ortaggio02}): 
\be
 \d s^2 =2\d u\left[\d r+\left(\lambda r^2+H(u,x,y)\right)\d u\right]+ P^{-2}(\d x^2+\d y^2) , \qquad P=1-k(x^2+y^2), 
 \label{KundtII}
\ee
with the 2-form field 
\be
 \bF= a\d u\wedge\d r+b{\bf V}_2+\d u\wedge \d \phi, \qquad \phi=\phi(u,x,y), 
 \label{F_KundtII}
\ee 
and ${\bf V}_2=P^{-2}\d x \wedge\d y$ 
is the volume form on the transverse space. Furthermore, $a$, $b$, $k$ and $\lambda$ are constants.  
Let us split the exterior derivative, $\d$, into derivatives over the two parts so that $\d=\d_1+\d_2$. Then we note $\d u\wedge \d \phi=\d u\wedge \d_2 \phi$. 
Clearly, $\d \bF =0$, and requiring $\d  * \bF=0$, implies $\Box_2\phi=0$, where $\Box_2 =- *_2\d_2 *_2 \d_2 $ ($ *_2$ is the Hodge dual on the transverse space) is the Laplacian on the transverse space (cf. also \cite{PodOrt03,Kadlecovaetal09}). 
\begin{proposition}
The 2-form field $\bF$ \eqref{F_KundtII} in the Kundt metric \eqref{KundtII}, where $\Box_2\phi=0$, is universal.
\end{proposition}

\begin{proof}

Let us start with some preliminary comments, to be used in the following. We introduce the null Kundt co-frame $\ell_a\d x^a=\d u$, $n_a\d x^a=(\lambda r^2+H)\d u+\d r$, $m^{(2)}_{a}\d x^a=P^{-1}\d x$, $m^{(3)}_{a}\d x^a=P^{-1}\d y$, so that $\d s^2=2\bl\bn+\bbm^{(2)}\bbm^{(2)}+\bbm^{(3)}\bbm^{(3)}$, and $\bF=a\bl\wedge\bn+b\bbm^{(2)}\wedge\bbm^{(3)}+\bl\wedge\d\phi$. We observe that \eqref{KundtII} is a degenerate Kundt spacetime and therefore the function $H$ does not enter the b.w. 0 components of any curvature tensors \cite{ColHerPel10}, which thus coincide with those of the symmetric space obtained by setting $H=0$ in \eqref{KundtII} (which are constant). Moreover, in the above frame, in addition to \eqref{Kundt}, one has  $L_{10}=\M{i}{j}{0}=N_{i0}=L_{1i}=L_{1i}=L_{i1}=\M{i}{j}{1}=N_{ij}=0$, $R_{101i}=R_{1ijk}=0$, $D\M{i}{j}{k}=DN_{i1}=D^2L_{11}=0$ (cf. \eqref{11n} and \eqref{11f}), $DR_{1i1j}=0$ (cf. \eqref{B4}) and $DF_{1i}=0$ (and $H$ enters only the  Ricci rotation coefficient $N_{i1}$). This implies that all the covariant derivatives of $\bF$ are {\em balanced} tensors (Definition~\ref{def_balanced} and Lemma~\ref{lemma_deriv}) and all the covariant derivatives of the Riemann tensor are {\em 1-balanced} tensors (Definition~\ref{def_1balanced} and Lemma~\ref{lemma_1deriv}).

 Let us consider an arbitrary 2-form $\tilde\bF$ constructed from $\bF$. Then it can be expressed as: 
\[\tilde\bF =A\bl\wedge\bn+B \bbm^{(2)}\wedge \bbm^{(3)}+\bl\wedge\left(C\d_2\phi+D *_2\d_2 \phi\right), \]
where $A, B, C, D$ are constants. Recalling $\Box_2\phi=0$, it is clear that $\tilde\bF$ obeys $\d\tilde\bF=\d *\tilde\bF=0$; hence, $\bF$ is 0-universal. 

As for $\tilde\bF$ from derivatives, we note that the boost-weight zero components of $\bF$ and of any curvature tensor are $O(1,1)\times O(2)$-symmetric. Moreover, the b.w. 0 component of $\nabla^I\bF$ {vanish}. 
This implies that any tensor $\tilde\bF$ must have $O(1,1)\times O(2)$-symmetric b.w.~0 components as well. Hence, 
\[(\tilde\bF)_0 =\hat A\bl\wedge\bn+\hat B \bbm^{(2)}\wedge \bbm^{(3)},\]
where $\hat A$ and $\hat B$ are constants,
and thus  $\d(\tilde\bF)_0=\d *(\tilde\bF)_0=0$. Next, {since} curvature tensors (including the derivatives) are of the form $R=(R)_0+(R)_{-2}+...$, the b.w. $-1$ component of $\tilde\bF$ needs to be of the form (recall the notation of \eqref{tF}): 
\[ 
 (\tilde\bF)_{-1}=\left[(R)_0\otimes (\bF)_0\otimes\ldots\otimes(\bF)_0\otimes(\bF+\nabla\bF+\ldots+\nabla^I\bF)_{-1}\right]_2. 
 \] 
Now, as observed above, all the b.w. zero components need to be $O(1,1)\times O(2)$-symmetric, and using the arguments in the proof of Theorem~\ref{prop_UVSI},  after the contraction, $(\tilde\bF)_{-1}$ is of the form: 
\[(\tilde\bF)_{-1} =\sum_I \bl\wedge\left(\hat C\d_2(\Box_2)^I\phi+\hat D *_2\d_2(\Box_2)^I\phi\right), \]
where $\hat C$ and $\hat D$ are constants,
which is zero if $I>0$, and if $I=0$, then $\d(\tilde\bF)_{-1}=\d *(\tilde\bF)_{-1}=0$. Hence, any $\tilde\bF$ constructed from $\bF$ obeys the Maxwell equations $\d\tilde\bF=\d *\tilde\bF=0$
\end{proof}

This shows that there are examples of universal Maxwell fields in four-dimensional type II Kundt spacetimes  and similar examples of universal $p-$forms are believed to exist in $n=2p$ dimensions {(cf. section~\ref{subsec_dirprod} for slightly different type II examples in higher dimensions)}. 
\section{Necessary conditions}

\label{sec_necessary}

\subsection{General dimension}

{Here we show that some necessary conditions can be obtained for a universal $\bF$ which is {\em not} CSI. Let us thus} assume there exists a non-constant invariant $I$, i.e., $I_{,a}\neq0$. Since $\bF$ is assumed to be universal, {\em any} $\tilde\bF$ constructed according to Definition~\ref{def_UM} must be divergenceless. Take, for instance, $\tilde\bF=I\bF$. Then one immediately obtains
\be
	I^{;a}F_{a b_1\ldots b_{p-1}}=0 .
	\label{IF}
\ee
Since $\bF$ is closed, this immediately implies that the Lie derivative of $\bF$ w.r.t. the gradient of $I$ is zero, i.e., $\pounds_{\nabla I}\bF=0$. This defines a symmetry of $\bF$.

In the special even dimensional case $n=2p$, one can also take $\tilde\bF=I{}^*\bF$, thus additionally obtaining
\be
	I^{;a}{}^*F_{a b_1\ldots b_{p-1}}=0 \qquad	(n=2p) .
	\label{I*F}
\ee
(And now also $\pounds_{\nabla I}{}^*\bF=0$.)

Eqs.~\eqref{IF} and \eqref{I*F} mean that, for $n=2p$, a universal $\bF$ which possesses a non-constant invariant $I$ must be null and aligned to the gradient of $I$, which in turn defines a null vector field (see Definition~\ref{def_N}) 
\be
 \ell_a\equiv I_{,a}  \qquad	(n=2p) .
\ee
It is also clear that all non-constant invariants must define the same null direction (since a null $p$-form admits a unique aligned null direction). Since all possible $\tilde\bF$ must be closed and co-closed when $n=2p$, one can similarly argue that they must also be null and aligned with $\bl$ -- this observation will be useful in the following. Note also that, since $\bl$ is null and a gradient, it is automatically geodesic and twistfree. One also has $D\hat I=0$ for any invariant $\hat I$ (and $\hat I=\hat I(u)$ if coordinates are chosen such that $\ell_a\d x^a=\d u$). We have thus arrived at
\begin{proposition}[Necessary conditions for a universal $\bF$ in the case $n=2p$]
\label{prop_univ_necess}
	In a spacetime of dimension $n=2p$, if a $p$-form $\bF$ is universal, then either (i) $\bF$ is CSI or (ii) $\bF$ (as well as any $\tilde\bF$ constructed according to Definition~\ref{def_UM}) is null and aligned with a geodesic and twistfree null direction, along which all the invariants of $\bF$ are constant.
\end{proposition}
From the above proof, it is clear that a similar statement holds replacing ``universal'' by ``$K$-universal'' and ``CSI'' by ``CSI$_K$''.

\subsection{Case $n=4=2p$: a universal $\bF$ must be CSI}

In the physically most important case $n=4=2p$, the non-CSI case cannot actually occurr, as we now show. By {\em reductio ad absurdum}, let us assume that we have a non-CSI universal $\bF$. By Proposition~\ref{prop_univ_necess}, $\bF$ is null and satisfies the Maxwell equations, thanks to which the Mariot-Robinson theorem \cite{Stephanibook} ensures that $\bl$ is also shearfree. This means that the permitted backgrounds can only be the Robinson-Trautman (if $\bl$ is expanding) or the Kundt (if $\bl$ is non-expanding) spacetimes.

\subsubsection{Robinson-Trautman}

\label{subsubsec_RT}

Let us first show that the Robinson-Trautman metrics are in fact forbidden if $\bF$ is universal and non-CSI. The idea is to construct from $\bF$ a 2-form $\tilde\bF$ which cannot be null (which would be required by universality) unless $\theta=0$, thus ending up in the Kundt family. Let us employ a frame ($i,j,\ldots=2,3$) parallelly transported along $\bl$. Let us denote $f_i\equiv F_{1i}$, so that $F_{ab}=2(f_2\ell_{[a}m^{(2)}_{b]}+f_3\ell_{[a}m^{(3)}_{b]})$. 
In the following, it will be understood that the following Maxwell equation (of b.w. 0)
\be
  Df_{i}=-\theta {f_{i}} ,
	\label{DMaxw_4D}
\ee
is employed. Thanks to \eqref{DMaxw_4D}, without loss of generality one can use an $r$-independent spin to set
\be
  f_3=0 .
\ee

Using this simplification, let us consider the 2-form
\be
	\tilde F_{ab}\equiv \left(F_{cd}^{\ \ \ ;e}F^{cd;f}F_{e[a|;f}\right)\left(F_{gh}^{\ \ \ ;p}F^{gh;q}{}^*F_{p|b];q}\right)  \approx 4(\theta f_2)^6m^{(2)}_{[a}m^{(3)}_{b]} ,
	\label{FFF}
\ee
where (from now on) the symbol $\approx$ indicates equality ``up to terms of lower b.w.''. 
As observed in Proposition~\ref{prop_univ_necess}, $\tilde F_{ab}$ has to be null, so $\theta=0$, as we wanted to prove. Therefore, in four dimensions there exist no universal non-CSI $2$-forms with an aligned expanding null direction.

\subsubsection{Kundt} 

Since $\bF$ is null and $\pounds_{\bl}\bF=0$ {(as pointed out above)}, we conclude that $\bF$ is VSI$_1$ (Proposition~C.1 of \cite{OrtPra16}), i.e., the only possible non-zero invariants must be constructed from second (or higher) covariant derivatives of $\bF$.

First, let us observe that in a Kundt spacetime one has (cf. {\eqref{dl}, \eqref{dm}})
\be
 \ell_{a;bc}\approx (DL_{1i} \ell_a \msub{i}{b}+DL_{i1} \msub{i}{a} \ell_b) n_c  \label{d2l} .
\ee
We can thus consider the 2-form 
\be
 \tilde F_{ab}\equiv\ell_{[a|;c|b]}\ell^{d;c}_{\ \ \ d}\approx (DL_{1i}DL_{1i})\ell_{[a}n_{b]} ,
\ee
which must be null and aligned with $\bl$, from which we conclude $DL_{1i}=0$ (recall that $\ell_a=I_{,a}$ and thus $\tilde F_{ab}$ is indeed constructed from $F_{ab}$). Similarly, taking instead $\tilde F_{ab}\equiv\ell_{c;[ab]}\ell^{c;d}_{\ \ \ d}$, we conclude that also $DL_{i1}=0$. Thanks to this, we have that $\ell_{a;bc}$ has components at most  of b.w. $-1$, while $\msub{i}{a ; bc }$ (thanks to {\eqref{11n}}) at most of b.w. 0, so that
\be
 \ell_{a;bcd}\approx D^2L_{11} \ell_a\ell_b n_c n_d  \label{d3l} .
\ee
It is now clear that by considering the 2-form
\be
 \tilde F_{ab}\equiv\ell_{[a|;c|b]}^{\ \ \ \ \ \ \ c}\approx D^2L_{11}\ell_{[a}n_{b]} ,
\ee
one concludes that $D^2L_{11}=0$. {With $DL_{i1}=0=D L_{1i}$}, this implies that the spacetime is degenerate Kundt (see again appendix~A of \cite{OrtPra16}) and therefore $\bF$ is VSI (Theorem~1.3 of \cite{OrtPra16}), i.e., all its invariant vanish and are thus constant, leading to a contradiction.

To summarize, we have shown that in four dimensions Proposition~\ref{prop_univ_necess} takes the following stronger form
\begin{theorem}[Necessary conditions for a universal $\bF$ in the case $n=4=2p$]
\label{prop_univ_necess_4D}
	In a four-dimensional spacetime, if a $2$-form is universal, then it is CSI.
\end{theorem}

We will now use this result to further constraint universal Maxwell fields in four dimensions, discussing separately the case of null and non-null fields. It will be convenient to employ the complex Newman-Penrose formalism, with the convention of \cite{Stephanibook}.\footnote{In the remaining part of this section, and only here, we will thus have $\ell_a n^a=-1$, instead of our usual convention $\ell_a n^a=+1$.} In a complex frame $(\bl,\bn,\mbox{\boldmath{$m$}},\mbox{\boldmath{$\bar m$}})$, such that $g_{ab}=2m_{(a}\bar m_{b)}-2\ell_{(a}n_{b)}$, the Maxwell equations read \cite{Stephanibook}
\beqn
 D\Phi_1-\bar\delta\Phi_0  & = & (\pi-2\alpha)\Phi_0+2\rho\Phi_1-\kappa\Phi_2 \label{max1} , \\
 D\Phi_2-\bar\delta\Phi_1 & = & -\lambda\Phi_0+2\pi\Phi_1+(\rho-2\varepsilon)\Phi_2 \label{max2} , \\
 \delta\Phi_1-\T\Phi_0 & = & (\mu-2\gamma)\Phi_0+2\tau\Phi_1-\sigma\Phi_2 \label{max3} , \\
 \delta\Phi_2-\T\Phi_1 & = & -\nu\Phi_0+2\mu\Phi_1+(\tau-2\beta)\Phi_2 \label{max4} .
\eeqn

\subsection{Case $n=4=2p$: null $\bF$}

\label{subsec_null4D}

In an adapted frame, the self-dual 2-form ${\cal F}_{ab}=F_{ab}+i{}^*F_{ab}$ takes the form \cite{Stephanibook}
\be
 {\cal F}_{ab}=4\Phi_2 \ell_{[a}m_{b]} ,
\label{F_null}
\ee
where $\bl$ is geodesic and shearfree by the Mariot-Robinson theorem, and $\Phi_2\neq0$. With no loss of generality we can choose an affine parametrization and a frame parallelly transported along $\bl$, so that 
\be
 \kappa=\sigma=\varepsilon=\pi=0 .
\ee

We can thus write the covariant derivatives of the frame vectors $\bl$ and $\mbox{\boldmath{$m$}}$ (up to terms of lower b.w.) as 
\beqn
 & & \ell_{a;b}\approx -\rho\bar m_a m_b-\bar\rho m_a \bar m_b , \\
 & & m_{a;b}\approx -\rho n_a m_b ,
\eeqn
while \eqref{max2} gives 
\be
 D\Phi_2=\rho\Phi_2 .
\ee
From these expressions and \eqref{F_null}, it follows
\be
 \big({\cal F}_{cd;e}{\cal F}^{cd}_{\ \ [;a|}\big)\big({\cal \bar F}_{fg;|b]}{\cal\bar F}^{fg;e}\big)\approx 256|\rho\Phi_2|^4 m_{[a}\bar m_{b]} . 
\ee
By adding to this 2-form its dual (times $i$), we end up with the following self-dual non-null 2-form
\be
 {\cal\tilde F}_{ab}\approx 4\tilde\Phi_1 (m_{[a}\bar m_{b]}-\ell_{[a}n_{b]}) , \qquad \tilde\Phi_1=64|\rho\Phi_2|^4 ,
 \label{F1_0}
\ee 
for which $\bl$ is a PND. By universality, ${\cal\tilde F}_{ab}$ must also solve Maxwell's equations, while Proposition~\ref{prop_univ_necess_4D} implies that $\tilde\Phi_1$ is a constant. Eq.~\eqref{max1} thus reduces to $0=2\rho\tilde\Phi_1$, which gives  
\be
 \rho=0 .
\ee
This means that $\bl$ is a Kundt vector field. The Ricci identities (cf. eqs. (7.21k), (7.21d), (7.21e) of \cite{Stephanibook}) thus give
\be
 D(\beta-\bar\alpha)=0 . 
\ee

Using the above results, $\ell_{a;b}$ contains only terms of b.w. $-1$ (or less) and $m_{a;b}$ of b.w. $0$ (or less), while for the second covariant derivatives we get
\beqn
 & & \ell_{a;bc}\approx -D{(\bar\beta+\alpha)}\ell_a m_bn_c-D\bar\tau m_a \ell_bn_c+\mbox{c.c.} , \\
 & & m_{a;bc}\approx -D\tau n_a \ell_bn_c .
\eeqn
Eq.~\eqref{max2} now implies $D\Phi_2=0$, which with the commutator (7.6b,\cite{Stephanibook}) also gives $D\delta\Phi_2=0$. We thus obtain
\be
 {\cal F}_{ab;c}^{\ \ \ \ c}\approx -4\Phi_2D\tau(m_{[a}\bar m_{b]}-\ell_{[a}n_{b]}) .
\ee 
Similarly as for \eqref{F1_0}, this 2-form must obey Maxwell's equation with $\Phi_2D\tau$ being a constant, which requires (cf. \eqref{max3})
\be
 D\tau=0 .
\ee
This implies that the spacetime is of aligned Riemann type II, while from the Bianchi and Ricci identities it further follows\footnote{For brevity, we refer to Appendix~A of \cite{OrtPra16} for more details (using footnote~4 of \cite{OrtPraPra07} for a dictionary between the complex and real Ricci rotation coefficients).}
\be
 D(\beta+\bar\alpha)=0 , \qquad D^2\mu=0 ,  \qquad D^2\lambda=0 , \qquad D^2(\gamma-\bar\gamma)=0 .
\ee 
Therefore, we arrive at
\be
 \ell_{a;bcd}\approx -D^2(\gamma+\bar\gamma)\ell_a\ell_b n_cn_d .
\ee
For our purposes, it suffices to further observe that $\ell_{a;b}$ and $\ell_{a;bc}$ possess only terms of negative b.w., while $m_{a;b}$, $m_{a;bc}$, $m_{a;bcd}$ and (using (7.6a,\cite{Stephanibook})) $\Phi_{2;abc}$ only terms of b.w. $0$ (or less). It is then easy to obtain 
\be
 ({\cal F}_{ab;d}^{\ \ \ \ d})_{;c}\approx 4\Phi_2D^2(\gamma+\bar\gamma)\ell_{[a} m_{b]}n_c .
\ee 
By universality, the 2-form ${\cal F}_{ab;d}^{\ \ \ \ d}$ must satisfy Maxwell's equations, which implies
\be
 D^2(\gamma+\bar\gamma)=0 .
\ee 

This means that the spacetime is degenerate Kundt, aligned with $\bl$. Thanks to \cite{OrtPra16}, we conclude that 
\begin{proposition}[Necessary conditions for a universal null $\bF$ in the case $n=4=2p$]
\label{prop_null_univ_necess_4D}
	In a four-dimensional spacetime, if a $2$-form is null and universal, then it is VSI.
\end{proposition}
{
This is clearly a specialization of Theorem~\ref{prop_univ_necess_4D} to the case of null fields, also constraining the background spacetime to be degenerate Kundt (aligned) \cite{OrtPra16}.}

\section*{Acknowledgments}

M.O. and V.P. have been supported by research plan RVO: 67985840 and by the Albert Einstein Center for Gravitation and Astrophysics, Czech Science Foundation GACR 14-37086G. The stay of M.O. at Instituto de Ciencias F\'{\i}sicas y Matem\'aticas, Universidad Austral de Chile has been supported by CONICYT PAI ATRACCI{\'O}N DE CAPITAL HUMANO AVANZADO DEL EXTRANJERO Folio 80150028. S.H. was supported through the Research Council of Norway, Toppforsk grant no. 250367: \emph{Pseudo-Riemannian Geometry and Polynomial Curvature Invariants:
	Classification, Characterisation and Applications.}

\renewcommand{\thesection}{\Alph{section}}
\setcounter{section}{0}

\renewcommand{\theequation}{{\thesection}\arabic{equation}}

\section{Kundt spacetimes and balanced tensors}
\setcounter{equation}{0}

\label{app_D_Kundt}

{In this appendix we review certain known properties of Kundt spacetimes and of balanced tensors useful in this paper. Some new facts are also proven (when a proposition is a summary of known results, this is indicated by including the corresponding reference in the text of the proposition).}

\subsection{Notation}

\label{app_sub_not}

In an $n$-dimensional spacetime, we employ a frame which consists of two null vectors $\bl\equiv{\mbox{\boldmath{$m_{(0)}$}}}$,  $\bn\equiv{\mbox{\boldmath{$m_{(1)}$}}}$ and $n-2$ orthonormal spacelike vectors $\bm{i} $, with $a, b\ldots=0,\ldots,n-1$ while $i, j  \ldots=2,\ldots,n-1$. For indices $i, j, \ldots$, it is not necessary to distinguish between subscripts and superscripts. 
The Ricci rotation coefficients are defined by \cite{Pravdaetal04}
\be
 L_{ab}=\ell_{a;b} , \qquad N_{ab}=n_{a;b}  , \qquad \M{i}{a}{b}=m^{(i)}_{a;b} ,
 \label{Ricci_rot}
\ee
and satisfy the identities 
\be 
 L_{0a}=N_{1a}=N_{0a}+L_{1a}=\M{i}{0}{a} + L_{ia} = \M{i}{1}{a}+N_{ia}=\M{i}{j}{a}+\M{j}{i}{a}=0 . 
 \label{const-scalar-prod}
\ee
Covariant derivatives along the frame vectors are 
\be
D \equiv \ell^a \nabla_a, \qquad \T\equiv n^a \nabla_a, \qquad \delta_i \equiv m^{(i)a} \nabla_a . 
 \label{covder}
\ee

\subsection{Kundt spacetimes}

Here we summarize some properties of Kundt spacetimes (see, e.g., Appendix~A of \cite{OrtPra16} for more details and references). 

\subsubsection{General Kundt spacetimes}

Kundt spacetimes are defined by the existence of a null vector field $\bl$ such that 
\be
 L_{i0}=0, \qquad L_{ij}=0  . 
 \label{Kundt}
\ee
Without loss of generality, one can use an affine parametrization and a frame parallelly transported along $\bl$, such that, additionally,  \cite{Pravdaetal04,Coleyetal04vsi,OrtPraPra07}
\be
 L_{10}=0, \qquad \M{i}{j}{0}=0, \qquad N_{i0}=0 . 
 \label{Kundt_rot}
\ee

The covariant derivatives of the frame vectors then read 
\beqn
& & \ell_{a ; b }=L_{11} \ell_a \ell_b +L_{1i} \ell_a \msub{i}{b}+L_{i1} \msub{i}{a} \ell_b  \label{dl} , \\
& & \msub{i}{a ; b }=-{N}_{i1} \ell_a \ell_b-{L}_{i1} n_a \ell_b-{N}_{ij} \ell_a \msub{j}{b}+\M{i}{j}{1} \msub{j}{a} \ell_b  + \M{i}{k}{l} \msub{k}{a} \msub{l}{b} , \label{dm} \\
& & n_{a ; b }=-L_{11} n_a \ell_b -L_{1i} n_a \msub{i}{b}+N_{i1} \msub{i}{a} \ell_b  + N_{ij} \msub{i}{a} \msub{j}{b} \label{dn} .
\eeqn

From \eqref{Kundt}, \eqref{Kundt_rot}, it follows
\be
 R_{0i0j}=0 , \qquad R_{0ijk}=0 . 
\label{RiemIb}
\ee

The Ricci identities (11b), (11e), (11n), (11a), (11j), (11m) and (11f) of \cite{OrtPraPra07} reduce to
\beqn
 & & DL_{1i}=-R_{010i} , \qquad DL_{i1}=-R_{010i} , \label{11be} \\
 & & D\M{i}{j}{k}=0 , \label{11n} \\
 & & DL_{11}=-L_{1i} L_{i1}-R_{0101}, \label{11a} \\
 & & DN_{ij}=-R_{0j1i} \label{11j} , \\
 & &  D\M{i}{j}{1}=-\M{i}{j}{k}L_{k1}-R_{01ij} , \label{11m} \\
 & &  DN_{i1}=-N_{ij}L_{j1}+R_{101i} \label{11f} . 
\eeqn

\subsubsection{Kundt spacetimes of aligned Riemann type II and degenerate Kundt spacetimes }

A generic Kundt spacetime is of Riemann type I (cf. \eqref{RiemIb}). For the subclass of {\em Riemann type II}, (i.e., assuming that $R_{010i}=0$), the Bianchi identities (B3), (B5), (B12), (B1), (B6) and (B4) of \cite{Pravdaetal04} take the simpler form
\beqn
 & & DR_{01ij}=0 , \qquad DR_{0i1j}=0 , \qquad DR_{ijkl}=0 , \label{B12} \\ 
 & & D R_{101i}-\delta_i R_{0101}=-R_{0101} L_{i1}-R_{01is} L_{s1}-R_{0i1s} L_{s1}  \label{B1} , \\
 & & D R_{1kij}+\delta_k R_{01ij}=R_{01ij}L_{k1}-2R_{0k1[i} L_{j]1}+R_{ksij} L_{s1}-2R_{01[i|s} \M{s}{|j]}{k} , \label{B6} \\
 & & D R_{1i1j}-\T R_{0j1i}-\delta_j R_{101i}=R_{0101} N_{ij}-R_{01is} N_{sj}+R_{0s1i} N_{sj}+R_{0j1s} \M{s}{i}{1}+R_{0s1i} \M{s}{j}{1}  \nonumber \\
 & &		\qquad\qquad\qquad\qquad\qquad\qquad\qquad		 {}+2R_{101i} L_{[1j]}+R_{1ijs} L_{s1}+R_{101s} \M{s}{i}{j} .   \label{B4}  
\eeqn

The {\em degenerate} Kundt metrics \cite{ColHerPel09a,Coleyetal09}, which are a subset of the Kundt metrics of Riemann type II, are defined by
\begin{definition}[Degenerate Kundt metrics \cite{ColHerPel09a,Coleyetal09}] 
	A Kundt spacetime is ``degenerate'' if the Kundt null direction $\bl$ is also a multiple null direction of the Riemann tensor and of its covariant derivatives of arbitrary order (which are thus all of aligned type~II, or more special). 
	\label{def_deg}
\end{definition}

It is useful to recall the following
\begin{proposition}[Conditions for degenerate Kundt metrics \cite{OrtPra16}]
\label{prop_Kundt_deg}
A Kundt spacetime is degenerate iff it is of aligned Riemann type II and $\ell^a R_{,a}=0$. A Kundt spacetime for which the tracefree part of the Ricci tensor is of aligned type III must be degenerate.
\end{proposition}

\subsection{Balanced tensors in degenerate Kundt spacetimes}

\label{app_subsec_balanced} 

Let us start by recalling the following
\begin{definition}[Balanced scalars and tensors \cite{Pravdaetal02,Coleyetal04vsi}]
\label{def_balanced}
	In a frame parallely transported along an affinely parameterized geodesic null vector field $\bl$, a scalar $\eta$ of b.w. $b$ under a constant boost is a ``balanced scalar'' if $D^{-b}\eta=0$ for $b<0$ and $\eta = 0$ for $b\geq0$. A ``balanced tensor'' is a tensor whose components are all balanced scalars.
\end{definition}

\begin{definition}[1-balanced scalars and tensors \cite{HerPraPra14}]
\label{def_1balanced}
	In a frame parallely transported along an affinely parameterized geodesic null vector field $\bl$, a scalar $\eta$ of b.w. $b$ under a constant boost is a ``1-balanced scalar'' if $D^{-b-1}\eta=0$ for $b<-1$ and $\eta = 0$ for $b\geq-1$. A tensor whose components are all 1-balanced scalars is a ``1-balanced tensor''.
\end{definition}
Clearly, a 1-balanced tensor is also balanced and possesses non-zero components only of b.w. $-2$ or lower.

Below we will need {the following} result of \cite{OrtPra16}
\begin{lemma}[Derivatives of balanced tensors in degenerate Kundt spacetimes \cite{OrtPra16}]
\label{lemma_deriv}
 In a degenerate Kundt spacetime, the covariant derivative of a balanced tensor is again a balanced tensor.
\end{lemma}

In a degenerate Kundt spacetime, employing an affine parameter and a parallelly transported frame, one has \cite{OrtPra16}
\beqn
 & & \T D - D \T = L_{11} D + L_{i1} \delta_i ,  \label{TD} \\
 & & \delta_i D - D \delta_i=L_{1i} D , \label{dD} \\
 & & DL_{1i}=0 , \qquad DL_{i1}=0 , \qquad D\M{i}{j}{k}=0 , \label{Ricci_D_0} \\
 & & D^2N_{ij}=0 , \qquad D^2\M{i}{j}{1}=0 , \qquad D^2L_{11}=0 , \qquad  D^3N_{i1}=0 . \label{Ricci_D_-1-2}
\eeqn 
The above equations suffice to readily extend Lemma~4.2 of \cite{HerPraPra14} to the following
\begin{lemma}[1-balanced scalars in degenerate Kundt spacetimes]
	\label{lemma_1balanced}
	In a degenerate Kundt spacetime, employing an affine parameter and a parallelly transported frame, if $\eta$ is a 1-balanced scalar of b.w. $b$, then all the following scalars (ordered by b.w.) are also 1-balanced: $D\eta$; $L_{1i}\eta$, $L_{i1}\eta$, $\M{i}{j}{k}		\eta$, $\delta_i\eta$; $L_{11}\eta$, $N_{ij}\eta$, $\M{i}{j}{1}\eta$, $\T \eta$; $N_{i1}\eta$.	
\end{lemma}
When taking the covariant derivative of a 1-balanced tensor in a degenerate Kundt spacetime, only the 1-balanced scalars encompassed by Lemma~\ref{lemma_1balanced} will appear (cf. \eqref{dl}--\eqref{dn}), which immediately leads to an extension of Lemma~4.3 of \cite{HerPraPra14}
\begin{lemma}[Derivatives of 1-balanced tensors in degenerate Kundt spacetimes]
\label{lemma_1deriv}
 In a degenerate Kundt spacetime, the covariant derivative of a 1-balanced tensor is again a balanced 1-tensor.
\end{lemma}

In a degenerate Kundt spacetime, the tensors $\nabla^{(I)}$(Riem) are generically of aligned type II \cite{ColHerPel09a,Coleyetal09}. However, with a restriction on the Weyl and Ricci tensors one can prove the following stronger condition (which extends Proposition~5.1 of \cite{HerPraPra14}). 
\begin{proposition}[]
 \label{prop_Kundt_III}
In a Kundt spacetime in which the Weyl tensor and the tracefree part of the Ricci tensor are both of aligned type III or more special (implying Kundt degenerate), the tensors $\nabla^{(I)}$(Riem) are balanced for any positive integer $I$, and therefore they are all of aligned type III (or more special).
\end{proposition}

\begin{proof}

The assumptions mean that the Riemann tensor has the form \eqref{riem_spec} {and that \eqref{dl}--\eqref{dn} hold (which will be understood from now on)}. The Bianchi equations (B.5) and (B.7) of \cite{Pravdaetal04} (with \eqref{Kundt}, \eqref{RiemIb}) thus give $DR=0=\delta_i R$ {(i.e., $D{\cal R}_0=0=\delta_i {\cal R}_0$)}. 
By Proposition~\ref{prop_Kundt_deg} it trivially follows that the considered spacetimes are degenerate, so that $\nabla$(Riem) is of aligned type II or more special. Moreover, the only possible non-zero b.w. 0 components of $\nabla\mbox{(Riem)}$ must be proportional to $D{\cal R}_{-1}$. 
However, the Bianchi equations \eqref{B1} and \eqref{B6} with \eqref{Ricci_D_0} give $D R_{101i}=0=D R_{1kij}$ (i.e., {$D{\cal R}_{-1}=0$}), so that $\nabla$(Riem) is necessarily of aligned type III (or more special). 

Next, we want to show it is also balanced. First, a simple counting of the  b.w.s of the available Riemann components and Ricci rotation coefficients reveals that the {frame} components of $[\nabla\mbox{(Riem)}]_{-1}$ must be linear combinations (with constant coefficients) of terms of the form {$\T{\cal R}_{0}$, $\delta_i{\cal R}_{-1}$, $D{\cal R}_{-2}$, or ${\cal R}_{-1}$} multiplied by $L_{1i}$ (or by $L_{i1}$ or $\M{i}{j}{k}$).\footnote{Terms proportional to ${\cal R}_{0}$ do not contribute to $[\nabla (\mbox{Riem} )]_{-1}$ since the (tensorial) coefficient of ${\cal R}_{0}$ in \eqref{riem_spec} is covariantly constant.}
	 Using \eqref{11j}, {\eqref{TD}--\eqref{Ricci_D_-1-2}} and \eqref{B4} (which gives $D^2R_{1i1j}=0$), it is thus easy to show that {the $D$-derivative of the frame components of b.w.~$-1$ of $\nabla (\mbox{Riem} )$ vanish.} Similarly, employing {\eqref{Ricci_D_-1-2}}, one can show that {the $D^{-b}$-derivative of the frame components of b.w.~$b$ of $\nabla (\mbox{Riem} )$ vanish also for $b=-2,-3,-4$ (while $[\nabla (\mbox{Riem} )]_{b}=0$ identically for $b<-4$),} so that $\nabla$(Riem) is a balanced tensor, as we wanted to prove. By Lemma~\ref{lemma_deriv} the proof is complete.

\end{proof} 

\begin{remark}
 From Proposition~\ref{prop_Kundt_III}, it follows that only the components  ${\cal R}_0$ can enter the scalar curvature invariants, and therefore such spacetimes are contained in the CSI class iff $\T R=0$. We also observe that $R\neq 0$ requires $L_{i1}\neq0$ (as follows immediately from the Ricci identity (11i) of \cite{OrtPraPra07}), while for $R=0$ the considered spacetimes coincide with the VSI class (in which case $L_{i1}\neq0$ and $L_{i1}=0$ are both possible \cite{Coleyetal04vsi,Coleyetal06}). If one  restricts the Ricci tensor to the type~D, then necessarily $R=$const (by the contracted Bianchi identities), so that the spacetimes become Einstein and Proposition~\ref{prop_Kundt_III} reduces to Proposition~5.1 of \cite{HerPraPra14}.  
\end{remark}

\providecommand{\href}[2]{#2}\begingroup\raggedright\endgroup

%
%
%
%
%
%

\begin{thebibliography}{10}

\bibitem{Mie12}
G.~Mie, {\it Grundlagen einer {T}heorie der {M}aterie},  {\em Ann. Physik} {\bf
  342} (1912) 511--534.

\bibitem{Born33}
M.~Born, {\it Modified field equations with a finite radius of the electron},
  {\em Nature} {\bf 132} (1933) 282.

\bibitem{BorInf34}
M.~Born and L.~Infeld, {\it Foundations of the new field theory},  {\em Proc.
  R. Soc. {\rm A}} {\bf 144} (1934) 425--451.

\bibitem{Bopp40}
F.~Bopp, {\it Eine lineare {T}heorie des {E}lektrons},  {\em Ann. Physik} {\bf
  430} (1940) 345--384.

\bibitem{Podolsky42}
B.~Podolsky, {\it A generalized electrodynamics. {P}art {I} -- {N}on-quantum},
  {\em Phys. Rev.} {\bf 62} (1942) 68--71.

\bibitem{HeiEul36}
W.~Heisenberg and H.~Euler, {\it Folgerungen aus der {D}iracschen {T}heorie des
  {P}ositrons},  {\em Z. Phys.} {\bf 98} (1936) 714--732.

\bibitem{Weisskopf36}
V.~Weisskopf, {\it {\"U}ber die {E}lektrodynamik des {V}akuums auf {G}rund der
  {Q}uantentheorie des {E}lektrons},  {\em Kong. Dan. Vid. Sel. Mat. Fys. Med.}
  {\bf 14} (1936), no.~6 1--39.

\bibitem{Euler36}
H.~Euler, {\it {\"U}ber die {S}treuung von {L}icht an {L}icht nach der
  {D}iracschen {T}heorie},  {\em Ann. Physik} {\bf 418} (1936) 398--448.

\bibitem{Schwinger51}
J.~S. Schwinger, {\it On gauge invariance and vacuum polarization},  {\em Phys.
  Rev.} {\bf 82} (1951) 664--679.

\bibitem{Tseytlin00}
A.~A. Tseytlin, {\it {B}orn-{I}nfeld action, supersymmetry and string theory},
  in {\em The Many Faces of the Superworld} (M.~Shifman, ed.), pp.~417--452.
\newblock World Scientific Publishing, Singapore, 2000.
\newblock \href{http://xxx.lanl.gov/abs/hep-th/9908105}{{\tt hep-th/9908105}}.

\bibitem{Born37}
M.~Born, {\it Th\'eorie non-lin\'eaire du champ \'electromagn\'etique},  {\em
  Ann. Inst. H. Poincar{\'e}} {\bf 7} (1937) 155--265.

\bibitem{Plebanski70}
J.~Pleba\'nski, {\em Lectures on non-linear electrodynamics}.
\newblock Nordita, Copenhagen, 1970.

\bibitem{Schroedinger35}
E.~Schr{\"{o}}dinger, {\it Contributions to {B}orn's new theory of the
  electromagnetic field},  {\em Proc. Roy. Soc. London Ser. A} {\bf 150} (1935)
  465--477.

\bibitem{Schroedinger43}
E.~Schr{\"{o}}dinger, {\it A new exact solution in non-linear optics
  (two-wave-system)},  {\em Proc. Roy. Irish Acad.} {\bf A49} (1943) 59--66.

\bibitem{Deser75}
S.~Deser, {\it Plane waves do not polarize the vacuum},  {\em J. Phys.~A} {\bf
  8} (1975) 1972--1974.

\bibitem{OrtPra18}
M.~Ortaggio and V.~Pravda, {\it Electromagnetic fields with vanishing quantum
  corrections},  {\em Phys. Lett. {\rm B}} {\bf 779} (2018) 393--395.

\bibitem{DefDesEsp10}
C.~Deffayet, S.~Deser, and G.~Esposito-Far\'ese, {\it Arbitrary $p$-form
  {G}alileons},  {\em Phys. Rev. {\rm D}} {\bf 82} (2010) 061501.

\bibitem{Stephanibook}
H.~Stephani, D.~Kramer, M.~MacCallum, C.~Hoenselaers, and E.~Herlt, {\em Exact
  Solutions of {E}instein's Field Equations}.
\newblock Cambridge University Press, Cambridge, second~ed., 2003.

\bibitem{Sokolowskietal93}
L.~M. Sokolowski, F.~Occhionero, M.~Litterio, and L.~Amendola, {\it Classical
  electromagnetic radiation in multidimensional space-times},  {\em Ann.
  Physics} {\bf 225} (1993) 1--47.

\bibitem{OrtPra16}
M.~Ortaggio and V.~Pravda, {\it Electromagnetic fields with vanishing scalar
  invariants},  {\em Class. Quantum Grav.} {\bf 33} (2016) 115010.

\bibitem{Milsonetal05}
R.~Milson, A.~Coley, V.~Pravda, and A.~Pravdov\'a, {\it Alignment and
  algebraically special tensors in {L}orentzian geometry},  {\em Int. J. Geom.
  Meth. Mod. Phys.} {\bf 2} (2005) 41--61.

\bibitem{OrtPraPra13rev}
M.~Ortaggio, V.~Pravda, and A.~Pravdov\'a, {\it Algebraic classification of
  higher dimensional spacetimes based on null alignment},  {\em Class. Quantum
  Grav.} {\bf 30} (2013) 013001.

\bibitem{Ficken39}
F.~A. Ficken, {\it The {R}iemannian and affine differential geometry of
  product-spaces},  {\em Ann. Math.} {\bf 40} (1939) 892--913.

\bibitem{Durkeeetal10}
M.~Durkee, V.~Pravda, A.~Pravdov\'a, and H.~S. Reall, {\it Generalization of
  the {G}eroch-{H}eld-{P}enrose formalism to higher dimensions},  {\em Class.
  Quantum Grav.} {\bf 27} (2010) 215010.

\bibitem{Ortaggio07}
M.~Ortaggio, {\it Higher dimensional spacetimes with a geodesic, shearfree,
  twistfree and expanding null congruence},  in {\em Proceedings of the XVII
  {SIGRAV} Conference (Torino, September 4--7, 2006)}, 2007.
\newblock \href{http://xxx.lanl.gov/abs/gr-qc/0701036}{{\tt gr-qc/0701036}}.

\bibitem{syngespec}
J.~L. Synge, {\em Relativity: the Special Theory}.
\newblock North-Holland, Amsterdam, 1955.

\bibitem{penrosebook2}
R.~Penrose and W.~Rindler, {\em Spinors and Space-Time}, vol.~2.
\newblock Cambridge University Press, Cambridge, 1986.

\bibitem{Bergmann42}
P.~G. Bergmann, {\em Introduction to the Theory of Relativity}.
\newblock Prentice-Hall, New York, 1942.

\bibitem{RobRoz84}
I.~Robinson and K.~R\'ozga, {\it Lightlike contractions on {M}inkowski
  space-time},  {\em J. Math. Phys.} {\bf 25} (1984) 499--505.

\bibitem{Guven87}
R.~G{\"u}ven, {\it Plane waves in effective theories of superstrings},  {\em
  Phys. Lett. {\rm B}} {\bf 191} (1987) 275--281.

\bibitem{Weitzenbock23}
R.~Weitzenb{\"{o}}ck, {\em Invariantentheorie}.
\newblock P. Noordhoff, Groningen, 1923.

\bibitem{Coleyetal04vsi}
A.~Coley, R.~Milson, V.~Pravda, and A.~Pravdov\'a, {\it Vanishing scalar
  invariant spacetimes in higher dimensions},  {\em Class. Quantum Grav.} {\bf
  21} (2004) 5519--5542.

\bibitem{ColHerPel06}
A.~Coley, S.~Hervik, and N.~Pelavas, {\it On spacetimes with constant scalar
  invariants},  {\em Class. Quantum Grav.} {\bf 23} (2006) 3053--3074.

\bibitem{Coleyetal07}
A.~Coley, A.~Fuster, S.~Hervik, and N.~Pelavas, {\it Vanishing scalar invariant
  spacetimes in supergravity},  {\em JHEP} {\bf 0705} (2007) 032.

\bibitem{KowalskiGlikman84}
J.~Kowalski-Glikman, {\it Vacuum states in supersymmetric {K}aluza-{K}lein
  theory},  {\em Phys. Lett. {\rm B}} {\bf 134} (1984) 194--196.

\bibitem{Hull84}
C.~Hull, {\it Exact pp-wave solutions of 11-dimensional supergravity},  {\em
  Phys. Lett. {\rm B}} {\bf 139} (1984) 39.

\bibitem{AmaKli89}
D.~Amati and C.~Klim\v{c}\'{\i}k, {\it Nonperturbative computation of the
  {W}eyl anomaly for a class of nontrivial backgrounds},  {\em Phys. Lett. {\rm
  B}} {\bf 219} (1989) 443--447.

\bibitem{HorSte90}
G.~T. Horowitz and A.~R. Steif, {\it Spacetime singularities in string theory},
   {\em Phys. Rev. Lett.} {\bf 64} (1990) 260--263.

\bibitem{Tseytlin93}
A.~A. Tseytlin, {\it String vacuum backgrounds with covariantly constant null
  {K}illing vector and two-dimensional quantum gravity},  {\em Nucl. Phys. {\rm
  B}} {\bf 390} (1993) 153--172.

\bibitem{BerKalOrt93}
E.~Bergshoeff, R.~Kallosh, and T.~Ortin, {\it Supersymmetric string waves},
  {\em Phys. Rev. {\rm D}} {\bf 47} (1993) 5444--5452.

\bibitem{Tod83}
K.~P. Tod, {\it All metrics admitting super-covariantly constant spinors},
  {\em Phys. Lett. {\rm B}} {\bf 121} (1983) 241--244.

\bibitem{Gauntlettetal03}
J.~P. Gauntlett, J.~B. Gutowski, C.~M. Hull, S.~Pakis, and H.~S. Reall, {\it
  All supersymmetric solutions of minimal supergravity in five dimensions},
  {\em Class. Quantum Grav.} {\bf 20} (2003) 4587--4634.

\bibitem{Ortinbook}
T.~Ort\'{\i}n, {\em Gravity and strings}.
\newblock Cambridge University Press, Cambridge, second~ed., 2015.

\bibitem{FigPap01}
J.~M. Figueroa-O'Farrill and G.~Papadopoulos, {\it Homogeneous fluxes, branes
  and a maximally supersymmetric solution of {M}-theory},  {\em JHEP} {\bf
  0108} (2001) 036.

\bibitem{GarAlvar84}
A.~Garc\'{\i}a~D. and M.~Alvarez~C., {\it Shear-free special electrovac
  type-{II} solutions with cosmological constant},  {\em Nuovo Cimento {\rm B}}
  {\bf 79} (1984) 266--270.

\bibitem{Khlebnikov86}
V.~I. Khlebnikov, {\it Gravitational radiation in electromagnetic universes},
  {\em Class. Quantum Grav.} {\bf 3} (1986) 169--173.

\bibitem{Lewand92}
J.~Lewandowski, {\it Reduced holonomy group and {E}instein equations with a
  cosmological constant},  {\em Class. Quantum Grav.} {\bf 9} (1992)
  L147--L151.

\bibitem{Ortaggio02}
M.~Ortaggio, {\it Impulsive waves in the {N}ariai universe},  {\em Phys. Rev.
  {\rm D}} {\bf 65} (2002) 084046.

\bibitem{PodOrt03}
J.~Podolsk\'y and M.~Ortaggio, {\it Explicit {K}undt type {$II$} and {$N$}
  solutions as gravitational waves in various type {$D$} and {$O$} universes},
  {\em Class. Quantum Grav.} {\bf 20} (2003) 1685--1701.

\bibitem{Kadlecovaetal09}
H.~Kadlecov{\'a}, A.~Zelnikov, P.~Krtou\v{s}, and J.~Podolsk\'{y}, {\it
  Gyratons on direct-product spacetimes},  {\em Phys. Rev. {\rm D}} {\bf 80}
  (2009) 024004.

\bibitem{ColHerPel10}
A.~Coley, S.~Hervik, and N.~Pelavas, {\it Lorentzian manifolds and scalar
  curvature invariants},  {\em Class. Quantum Grav.} {\bf 27} (2010) 102001.

\bibitem{OrtPraPra07}
M.~Ortaggio, V.~Pravda, and A.~Pravdov\'a, {\it Ricci identities in higher
  dimensions},  {\em Class. Quantum Grav.} {\bf 24} (2007) 1657--1664.

\bibitem{Pravdaetal04}
V.~Pravda, A.~Pravdov\'a, A.~Coley, and R.~Milson, {\it Bianchi identities in
  higher dimensions},  {\em Class. Quantum Grav.} {\bf 21} (2004) 2873--2897.
  See also V. Pravda, A. Pravdov\'a, A. Coley and R. Milson {\em Class. Quantum
  Grav.} {\bf 24} (2007) 1691 (corrigendum).

\bibitem{ColHerPel09a}
A.~Coley, S.~Hervik, and N.~Pelavas, {\it Spacetimes characterized by their
  scalar curvature invariants},  {\em Class. Quantum Grav.} {\bf 26} (2009)
  025013.

\bibitem{Coleyetal09}
A.~Coley, S.~Hervik, G.~O. Papadopoulos, and N.~Pelavas, {\it Kundt
  spacetimes},  {\em Class. Quantum Grav.} {\bf 26} (2009) 105016.

\bibitem{Pravdaetal02}
V.~Pravda, A.~Pravdov\'a, A.~Coley, and R.~Milson, {\it All spacetimes with
  vanishing curvature invariants},  {\em Class. Quantum Grav.} {\bf 19} (2002)
  6213--6236.

\bibitem{HerPraPra14}
S.~Hervik, V.~Pravda, and A.~Pravdov\'a, {\it Type {III} and {N} universal
  spacetimes},  {\em Class. Quantum Grav.} {\bf 31} (2014) 215005.

\bibitem{Coleyetal06}
A.~Coley, A.~Fuster, S.~Hervik, and N.~Pelavas, {\it Higher dimensional {VSI}
  spacetimes},  {\em Class. Quantum Grav.} {\bf 23} (2006) 7431--7444.

\end{thebibliography}

\end{document}